\documentclass[a4paper,12pt]{article}
\usepackage[utf8]{inputenc}
\usepackage[english]{babel}
\usepackage{graphicx}
\usepackage{mathptmx}
\usepackage{enumitem}
\usepackage[onehalfspacing]{setspace}
\usepackage[headheight=1in,margin=1in]{geometry}
\usepackage{lipsum}

\usepackage[utf8]{inputenc}
\usepackage{textgreek}
\DeclareUnicodeCharacter{03B5}{\textepsilon}

\usepackage{amsmath,amssymb,amsthm,amsfonts}
\usepackage[authoryear,sort]{natbib}
\usepackage{titlesec}
\titlespacing{\paragraph}{0pt}{0pt}{1em}

\usepackage{subcaption}
\usepackage{float}

\usepackage{tikz}
\usetikzlibrary{positioning, arrows.meta}

\usepackage[hidelinks]{hyperref}

\usepackage[bottom,hang,flushmargin]{footmisc}
\usepackage{xcolor}
\definecolor{sunset}{RGB}{238,93,108}

\newcommand{\PP}{\mathbb{P}}
\newcommand{\EE}{\mathbb{E}}

\graphicspath{{figures/}}

\newtheoremstyle{customassump}
  {}{}                          
  {\itshape}                   
  {}                           
  {\bfseries}                  
  {}                           
  {.5em}                       
  {}                           
\theoremstyle{customassump}
\newtheorem{assumption}{Assumption}
\newtheorem{theorem}{Theorem}

\newtheorem{proposition}{Proposition}
\newtheorem{lemma}{Lemma}
\newtheorem*{proposition*}{Proposition}

\title{\vspace{-1in}\singlespacing Anytime-valid Optimal Policy Identification}

\author{Daniel Molitor\thanks{Department of Information Science, Cornell University, djm484@cornell.edu.}}
 
\date{}

\begin{document}

\maketitle

\begin{singlespacing}
\begin{abstract}
We develop an anytime-valid framework for optimal policy identification from logged contextual bandit data. In many applied settings, the analyst wants to select the optimal policy from a candidate policy class \(\Pi\), but data are generated by an externally determined logging policy that they do not control. The analyst may also wish to monitor evidence continuously, stopping once the optimal policy is clear rather than committing to a fixed sample size in advance. This paper addresses these challenges by constructing a time-indexed set \(S_t\) that retains the true optimal policy set uniformly over time with high probability. The resulting procedure allows the analyst to monitor policy values, eliminate clearly suboptimal policies, and stop at data-dependent times without invalidating inference. When the optimal policy is unique, we define a stopping time for its identification and derive a sample-complexity bound scaling as \(
O\!\left(\frac{\log |\Pi|+\log\log(1/\Delta_{\min})}{\Delta_{\min}^2}\right)
\),
where \(\Delta_{\min}\) is the gap between the best and second-best policy values. Simulations demonstrate that the anytime-valid approach can yield substantial sample savings relative to fixed-\(N\) designs. An application to a large adaptive experiment on reducing misinformation online illustrates how the method provides a dynamic view as evidence on the optimal policy accumulates.
\end{abstract}
\end{singlespacing}


\newpage

\section{Introduction}

Personalized treatment assignment policies are increasingly used to improve outcomes in experimental settings where individuals respond differently to the same interventions.
\citet{doss_more_2019} show that text messages to parents, personalized to their child's developmental level, can produce meaningful educational gains; \citet{offer-westort_battling_2024} find that
personalized nudges to consider the accuracy of online information
substantially reduce sharing of misinformation on social media.
In these examples, and in a growing body of work across education, public health, and digital platforms, tailoring treatment assignment policies to individual characteristics yields measurably better outcomes than non-personalized alternatives.
In many settings, these personalized policies are learned from sequentially collected data, with the policy evolving over time as evidence accumulates.

The contextual bandit framework is a canonical abstraction for these settings \citep{li_contextual-bandit_2010}. Informally, in this framework a decision-maker observes contextual information \(X_t \in \mathcal{X}\) for subject \(t\) (such as country, age, gender, etc.), chooses an action \(A_t \in \mathcal{A}\) (such as which online safety nudge to provide), and then observes a reward \(R_t\) (such as whether the subject shares misinformation on social media). A
policy \(\pi(a \mid x)\) is simply a conditional distribution over actions given a user's context \(X_t\). This paper addresses a question that naturally arises when a finite collection of candidate policies \(\Pi = \{\pi_1, \dots, \pi_m\}\) is under consideration: \emph{which policy is optimal?} Formally, letting \(\nu(\pi) := \mathbb{E}_\pi[R]\) denote the expected reward of a policy \(\pi\), the goal is to identify \(\pi^\star := {\arg\max}_{\pi \in \Pi}\, \nu(\pi)\).
A substantial literature studies optimal policy identification under adaptive sampling, where data collection is actively steered toward identifying the optimal policy \citep{li_instance-optimal_2022, zanette_design_2021}. These approaches build on the broader pure-exploration and best-arm identification literature, especially in structured and linear settings \citep{garivier_optimal_2016, soare_best-arm_2014, fiez_sequential_2019, jedra_optimal_2020, degenne_gamification_2020, tirinzoni_elimination_2022}.
This line of work assumes that the analyst controls data collection and can steer it toward the goal of optimal policy identification.

In many practical settings, however, this may not be the case. A digital platform may be running a production bandit algorithm tuned to its own engagement objective; a health system may operate under a treatment assignment protocol subject to regulatory or contractual constraints; a research team may be limited to using data generated by a separate team whose data collection process it cannot modify. In each case, the analyst is downstream of an externally determined policy \(h\) that generates the observed data---the logging policy---which may range from a sophisticated contextual bandit
to simple fixed-probability randomization. The analyst must then identify \(\pi^\star\) from the data stream generated by \(h\), without the ability to intervene
on \(h\). This task is inherently off-policy: the learning problem is no longer how to optimally collect data to identify \(\pi^\star\), but how to identify \(\pi^\star\) using the data collected under \(h\).

Additionally, much of the optimal policy identification and off-policy evaluation literature focuses on fixed-\(N\) or post-hoc approaches to policy evaluation. The resulting tools can estimate policy values \(\nu(\pi)\) and identify optimal policies \(\pi^\star\), but their inferential guarantees hold only at a fixed, prespecified sample size \citep{li_instance-optimal_2022, zanette_design_2021, dudik_doubly_2011, karampatziakis_off-policy_2021, bibaut_post-contextual-bandit_2021, hadad_confidence_2021, lee_off-policy_2024}.
From a practitioner’s perspective, however, this can be restrictive. In many applications, an analyst would like to monitor policy values as data accrue, progressively eliminate clearly suboptimal policies, and continue collecting data only when the evidence remains ambiguous.
A valid inferential framework for this goal must remain correct not only at a prespecified sample size, but uniformly over all sample sizes, including arbitrary data-dependent stopping times.
This is precisely the framework offered by anytime-valid inference \citep{howard_time-uniform_2021, waudby-smith_anytime-valid_2024}, which provides inferential guarantees that remain valid throughout the data collection process.

Building on the anytime-valid off-policy evaluation framework of \citet{waudby-smith_anytime-valid_2024}, we develop inferential tools for optimal policy identification under an externally determined logging policy. The resulting guarantees on the optimal policy set are nonasymptotic, nonparametric, and uniformly valid over all sample sizes. Analysts can therefore monitor evidence as it accrues, progressively eliminate clearly suboptimal policies, and stop at arbitrary data-dependent times without invalidating inference. The framework accommodates a broad class of stochastic logging policies, from fixed randomization to adaptive contextual bandit algorithms.

\subsection{Contributions}
The paper makes two contributions.

\begin{enumerate}[wide, labelindent=0pt]

\item We construct an anytime-valid optimal policy candidate set \(S_t \subseteq \Pi\), which at each time \(t\) retains the policies that remain plausibly optimal given the evidence observed so far. We show that \(S_t\) satisfies the uniform coverage guarantee
\[
\PP\!\left( \Pi^\star \subseteq S_t \;\text{ for all } t \geq 1 \right) \geq 1-\alpha,
\]
where \(\Pi^\star := \left\{\pi\in\Pi:\nu(\pi)=\max_{\pi'\in\Pi}\nu(\pi')\right\}\) is the set of optimal policies and \(\alpha\) is a user-chosen significance level. That is, with high probability simultaneously over all \(t\), \(S_t\) retains every optimal policy while progressively eliminating suboptimal policies as data accumulate.

\item When the optimal policy \(\pi^\star\) is unique, we define a data-driven stopping time \(\tau\) at which the optimal policy candidate set \(S_t\) uniquely identifies \(\pi^\star\). We show that this stopping time occurs with high probability and derive a corresponding sample-complexity bound that scales with \(\log |\Pi|\), \(\log\log \Delta_{\min}^{-1}\), and \(\Delta_{\min}^{-2}\), where \(\Delta_{\min} := \min_{\pi \neq \pi^\star}\{\nu(\pi^\star)-\nu(\pi)\}\) denotes the minimum suboptimality gap.

\end{enumerate}

\subsection{Related work}

\paragraph{Best-arm and policy identification}

A large literature studies how to adaptively collect data to identify the best arm or policy. Foundational work develops fixed-confidence guarantees for best-arm identification \citep{garivier_optimal_2016, soare_best-arm_2014, degenne_gamification_2020, tirinzoni_elimination_2022}. Closely related work extends these ideas to identification in structured, linear, and contextual settings \citep{fiez_sequential_2019, jedra_optimal_2020, zanette_design_2021, li_instance-optimal_2022}.
The setting considered here differs in a central respect: the analyst does not control the logging policy and must instead identify the optimal policy from logged data generated by an externally determined data-collection process.

\paragraph{Off-policy evaluation}

Off-policy evaluation studies how to estimate the value of a target policy using data collected under a different logging policy, with contextual bandits serving as a primary setting for such questions. Within this literature, closely related work includes doubly robust and adaptive approaches to off-policy evaluation \citep{dudik_doubly_2011, wang_optimal_2017}, off-policy selection using confidence bounds at a fixed horizon \citep{kuzborskij_confident_2022}, and statistical inference under adaptive data collection \citep{hadad_confidence_2021, bibaut_post-contextual-bandit_2021, zhang_statistical_2021, lee_off-policy_2024, leiner_adaptive_2026}. These papers address the inferential challenges of logged data, but typically in a fixed-\(N\) framework, where validity is tied to a prespecified sample size. The setting considered here differs: we seek to identify the optimal policy under sequential monitoring, with guarantees that remain valid uniformly over time.

\paragraph{Anytime-valid inference and off-policy evaluation}

This work is also closely related to the broader literature on anytime-valid inference, confidence sequences, and sequential experimentation \citep{howard_time-uniform_2021, johari_always_2022, lindon_anytime-valid_2025, molitor_anytime-valid_2025}. In contextual bandits, \citet{karampatziakis_off-policy_2021} and \citet{waudby-smith_anytime-valid_2024} develop confidence sequences for off-policy evaluation. We build directly on this line of work, but shift the inferential target from individual policy values to optimal policy identification. This yields a procedure for anytime-valid optimal policy identification from logged contextual bandit data under an externally determined logging policy.

\section{Setup and background}

We begin by formalizing the contextual bandit setup, the off-policy quantities of interest, and the confidence sequences that underpin our identification procedure. Much of the notation in this section follows \citet{waudby-smith_anytime-valid_2024},
tailored to our setting.

\subsection{Contextual bandit setup}

Let \(\mathcal X\) denote the context space and \(\mathcal A\) the action space.
We observe a sequential stream of contextual bandit data
\((X_t,A_t,R_t)_{t\ge1}\),
where at each time \(t\), a context \(X_t \in \mathcal X\) is observed, an action \(A_t \in \mathcal A\) is assigned, and a reward \(R_t\) is realized. Throughout, we assume that \(R_t\in[0,1]\), a standard assumption in contextual bandits \citep{karampatziakis_off-policy_2021, thomas_high-confidence_2015}. We write
\(
\mathcal H_t := \sigma\!\big((X_i,A_i,R_i)_{i=1}^t\big)
\)
for the natural filtration generated by the observed data up to time \(t\).

Data are collected under a \emph{predictable} sequence of logging policies \((h_t)_{t\ge1}\). Formally, conditional on the current context \(X_t\) and past history \(\mathcal H_{t-1}\),
\[
A_t \mid (X_t,\mathcal H_{t-1}) \sim h_t(\cdot \mid X_t),
\]
where \(h_t\) is \(\mathcal H_{t-1}\)-measurable. Thus, the logging policy may adapt over time as a function of the observed history, with fixed logging policies arising as a special case. Let \(\Pi=\{\pi_1,\dots,\pi_m\}\) denote a finite class of candidate target policies, with \(m=|\Pi|\), where each \(\pi \in \Pi\) is a conditional distribution over actions given context. Our goal is to identify the value-maximizing policy in \(\Pi\).

\subsection{Policy values and the target of inference}

For each target policy \(\pi\in\Pi\), define its conditional value at time \(t\) by
\(
\nu_t(\pi):=\EE_{\pi}[R_t \mid \mathcal H_{t-1}],
\)
where the expectation is taken under the counterfactual intervention that draws \(A_t\) from \(\pi(\cdot\mid X_t)\) rather than from the logging policy \(h_t(\cdot\mid X_t)\). Throughout, we work in the \emph{time-invariant policy value} setting: for each \(\pi\in\Pi\), there exists a constant \(\nu(\pi)\in[0,1]\) such that \(\nu_t(\pi)=\nu(\pi)\) for all \(t\ge1\).
As above, let \(\Pi^\star\) denote the set of optimal policies in \(\Pi\). When the optimal policy is unique we write
\(\pi^\star := \arg\max_{\pi\in\Pi}\nu(\pi)\).
Thus, identifying \(\pi^\star\) reduces to obtaining reliable inference on the collection of policy values \(\{\nu(\pi):\pi\in\Pi\}\) using data generated under \((h_t)_{t\ge1}\).

\subsection{Importance weights and doubly robust pseudo-outcomes}

To estimate \(\nu(\pi)\) from data generated under \((h_t)_{t\ge1}\), we rely on the importance weights
\[
w_t(\pi):=\frac{\pi(A_t\mid X_t)}{h_t(A_t\mid X_t)},
\]
which give the ratio of the realized action's probability under the target policy \(\pi\) to its probability under the logging policy. We assume that each candidate policy \(\pi\in\Pi\) is absolutely continuous with respect to \(h_t\), so that \(w_t(\pi)\) is almost surely finite.

Following \citet{waudby-smith_anytime-valid_2024}, we combine these weights with a reward model to construct doubly robust pseudo-outcomes. Let \(\widehat r_t(x,a)\in[0,1]\) denote any predictive model of \(R_t\) built from the past history \(\mathcal H_{t-1}\), and let \(k\in[0,\infty)\) be a tuning parameter that controls the truncation of the weighted reward model fitted values. For each target policy \(\pi\in\Pi\), define the lower and upper doubly robust pseudo-outcomes by
\begin{align}
\phi_t^{\mathrm{DRL}}(\pi)
&:= w_t(\pi)\!\left(R_t-\Big[\widehat r_t(X_t,A_t)\wedge \tfrac{k}{w_t(\pi)}\Big]\right) \notag \\
&\quad + \EE_{a\sim\pi(\cdot\mid X_t)}
\!\left[\widehat r_t(X_t,a)\wedge \tfrac{k}{w_t(\pi)}\right], \notag \\
\phi_t^{\mathrm{DRU}}(\pi)
&:= w_t(\pi)\!\left((1-R_t)-\Big[(1-\widehat r_t(X_t,A_t))\wedge \tfrac{k}{w_t(\pi)}\Big]\right) \notag \\
&\quad + \EE_{a\sim\pi(\cdot\mid X_t)}
\!\left[(1-\widehat r_t(X_t,a))\wedge \tfrac{k}{w_t(\pi)}\right] \notag.
\end{align}
These pseudo-outcomes satisfy
\[
\EE\!\left[\phi_t^{\mathrm{DRL}}(\pi)\mid \mathcal H_{t-1}\right]=\nu(\pi),
\qquad
\EE\!\left[\phi_t^{\mathrm{DRU}}(\pi)\mid \mathcal H_{t-1}\right]=1-\nu(\pi).
\]
Crucially, this does \emph{not} require \(\widehat r_t\) to be correctly specified: the reward model can improve efficiency, but validity follows from the known logging probabilities and importance weights. We now use these pseudo-outcomes to construct anytime-valid confidence sequences for the policy values \(\nu(\pi)\).

\subsection{Policy value confidence sequences}
\label{sec:conf_sequences}

A sequence of random intervals \([L_t,U_t]_{t\ge1}\) is called a \((1-\alpha)\) confidence sequence (CS) for a parameter \(\theta\) if
\(
\PP\!\left(\theta\in[L_t,U_t]\ \forall t\ge1\right)\ge 1-\alpha.
\)
To carry out inference on the policy values \(\nu(\pi)\), we use the CSs of \citet{waudby-smith_anytime-valid_2024}. Fix a target policy \(\pi\) and a truncation level \(k\ge0\). Let \(\phi_t\) denote either \(\phi_t^{\mathrm{DRL}}(\pi)\) or \(\phi_t^{\mathrm{DRU}}(\pi)\), depending on whether a lower or upper confidence bound is being constructed. Define the scaled pseudo-outcomes \(\xi_t:=\phi_t/(1+k)\) and the associated variance process \((V_t)_{t=1}^\infty\) by
\begin{equation}
V_t(\pi):=\sum_{i=1}^t\big(\xi_i-\widehat\xi_{i-1}\big)^2, \qquad
\text{where} \qquad
\widehat\xi_t:=\left(\frac1t\sum_{i=1}^t\xi_i\right)\wedge \frac{1}{1+k}.
\label{eq:variance_process}
\end{equation}
Here \(\widehat\xi_0\in\big[0,(1+k)^{-1}\big]\) is chosen by the analyst.

\begin{proposition}[{\citet[Proposition 3]{waudby-smith_anytime-valid_2024}}]
\label{prop:lil_cs}
Let \((X_t,A_t,R_t)_{t=1}^\infty\) be an infinite sequence of contextual bandit data with rewards in \([0,1]\), generated by the logging policies \((h_t)_{t=1}^\infty\). Let \((V_t)_{t=1}^\infty\) be the variance process defined in \eqref{eq:variance_process}, and define the stabilized process \(\overline V_t(\pi):=V_t(\pi)\vee 1\). Define the function \(\ell_t(\alpha)\)
\[
\ell_t(\alpha):=2\log(\log \overline V_t(\pi)+1)+\log\!\left(\frac{1.65}{\alpha}\right).
\]
Then
\[
L_t(\pi;\alpha)
:=
(1+k)\left(
\frac1t\sum_{i=1}^t\xi_i
-
\frac{\sqrt{2.13\,\ell_t(\alpha)\,\overline V_t(\pi)+1.76\,\ell_t(\alpha)^2}}{t}
-
\frac{1.33\,\ell_t(\alpha)^2}{t}
\right)\vee 0
\]
forms a lower \((1-\alpha)\)-CS for \(\nu(\pi)\), meaning
\(
\PP\!\left(\nu(\pi)\ge L_t(\pi;\alpha)\ \forall t\ge1\right)\ge 1-\alpha
\).
An analogous upper CS \((U_t)_{t=1}^\infty\) is obtained by mirroring. That is,
\(
U_t := 1 - L_t
\),
with \(\phi_t^{\mathrm{DRL}}(\pi)\) replaced by \(\phi_t^{\mathrm{DRU}}(\pi)\) in all preceding quantities.
\end{proposition}

\noindent Combining the corresponding lower and upper one-sided CSs via a union bound
(with error level \(\alpha/2\) in each tail), we obtain a two-sided
\((1-\alpha)\)-CS for \(\nu(\pi)\), denoted
\[
\left[L_t(\pi;\alpha),\,U_t(\pi;\alpha)\right]_{t\ge1}.
\]

\section{Anytime-valid optimal policy identification}

\subsection{Constructing an optimal policy set}

Using the policy-value CSs introduced in the previous section, we next construct an anytime-valid set estimator of the optimal policy set \(\Pi^\star\).

Fix a probability threshold \(\alpha \in (0,1)\), and for each \(\pi \in \Pi\) let
\(
\left[L_t\!\left(\pi;\frac{\alpha}{m}\right),\ U_t\!\left(\pi;\frac{\alpha}{m}\right)\right]_{t\ge1}
\)
denote a valid \((1-\alpha/m)\) CS for \(\nu(\pi)\). By a union bound over \(\Pi\), these intervals satisfy the simultaneous coverage guarantee
\begin{equation}
\mathbb P\!\left(
\nu(\pi)\in
\left[L_t\!\left(\pi;\frac{\alpha}{m}\right),\ U_t\!\left(\pi;\frac{\alpha}{m}\right)\right]
\ \forall t\ge1,\ \forall \pi\in\Pi
\right)
\ge 1-\alpha.
\label{eq:simultaneous-cs-coverage}
\end{equation}

This simultaneous coverage suggests a natural elimination rule. At time \(t\), a policy \(\pi\) should remain under consideration if its upper confidence bound is still at least as large as the largest lower confidence bound attained by any competing policy. Equivalently, if another policy's lower bound exceeds \(\pi\)'s upper bound, then with high probability \(\pi\) is a suboptimal policy. The following result formalizes this intuition.

\begin{theorem}[Anytime-valid optimal policy set]
\label{thm:optimal-policy-set}
For each \(t\ge1\), define the optimal policy candidate set
\[
S_t
:=
\left\{
\pi\in\Pi:
U_t\!\left(\pi;\frac{\alpha}{m}\right)
\ge
\max_{\pi'\in\Pi}
L_t\!\left(\pi';\frac{\alpha}{m}\right)
\right\}.
\]
Then
\[
\mathbb P\!\left(
\Pi^\star \subseteq S_t\ \forall t\ge1
\right)
\ge 1-\alpha.
\]
\end{theorem}

The candidate set \(S_t\) may be viewed as a sequential elimination procedure, as summarized in Figure~\ref{fig:method_flowchart}. As data accumulate, policies whose upper confidence bounds fall below the best lower confidence bound are eliminated. A proof can be found in Appendix~\ref{app:proof-optimal-policy-set}. Theorem~\ref{thm:optimal-policy-set} shows that, with probability at least \(1-\alpha\), no optimal policy is ever eliminated. This holds uniformly over time, yielding an anytime-valid candidate set for the optimal policy set.

\begin{figure}[H]
\centering
\makebox[\textwidth][c]{%
\begin{tikzpicture}[
    node distance=1.6cm and 1.1cm,
    >=Latex,
    box/.style={
        draw,
        rounded corners,
        thick,
        align=center,
        minimum height=1.5cm,
        text width=4.8cm,
        inner sep=7pt
    },
    line/.style={thick, ->},
    every node/.style={font=\small}
]

\node[box] (policies) {
    \textbf{Candidate policies}\\
    Define finite policy class\\
    \(\Pi = \{\pi_1,\ldots,\pi_m\}\)
};

\node[box, right=of policies] (data) {
    \textbf{Logged bandit data}\\
    \((X_t, A_t, R_t)_{t \ge 1}\)\\
    collected under logging policies \((h_t)_{t\ge1}\)
};

\node[box, right=of data] (dr) {
    \textbf{DR pseudo-outcomes}\\
    For each \(\pi \in \Pi\), construct\\
    \(\phi_t^{\mathrm{DRL}}(\pi)\) and \(\phi_t^{\mathrm{DRU}}(\pi)\)
};

\node[box, below=of dr] (cs) {
    \textbf{Policy value CSs}\\
    For each policy \(\pi\), construct\\ CS
    \([L_t(\pi),\, U_t(\pi)]\) for \(\nu(\pi)\)
};

\node[box, left=of cs] (st) {
    \textbf{Optimal policy candidate set}\\
    \(S_t = \{\pi \in \Pi :\)\\
    \(U_t(\pi) \ge \max_{\pi' \in \Pi} L_t(\pi')\}\)
};

\node[box, left=of st] (stop) {
    \textbf{Stopping rule}\\
    Stop when \(S_t\) is a singleton\\
    and output the remaining policy
};

\draw[line] (policies) -- (data);
\draw[line] (data) -- (dr);
\draw[line] (dr) -- (cs);
\draw[line] (cs) -- (st);
\draw[line] (st) -- (stop);

\end{tikzpicture}%
}
\caption{
\textbf{Sequential elimination procedure.} Construct policy value CSs from logged contextual bandit data. Use these CSs to construct the candidate set \(S_t\). Policies are eliminated when their upper confidence bounds fall below the best lower confidence bound; once \(S_t\) becomes a singleton, the remaining policy is selected.
}
\label{fig:method_flowchart}
\end{figure}
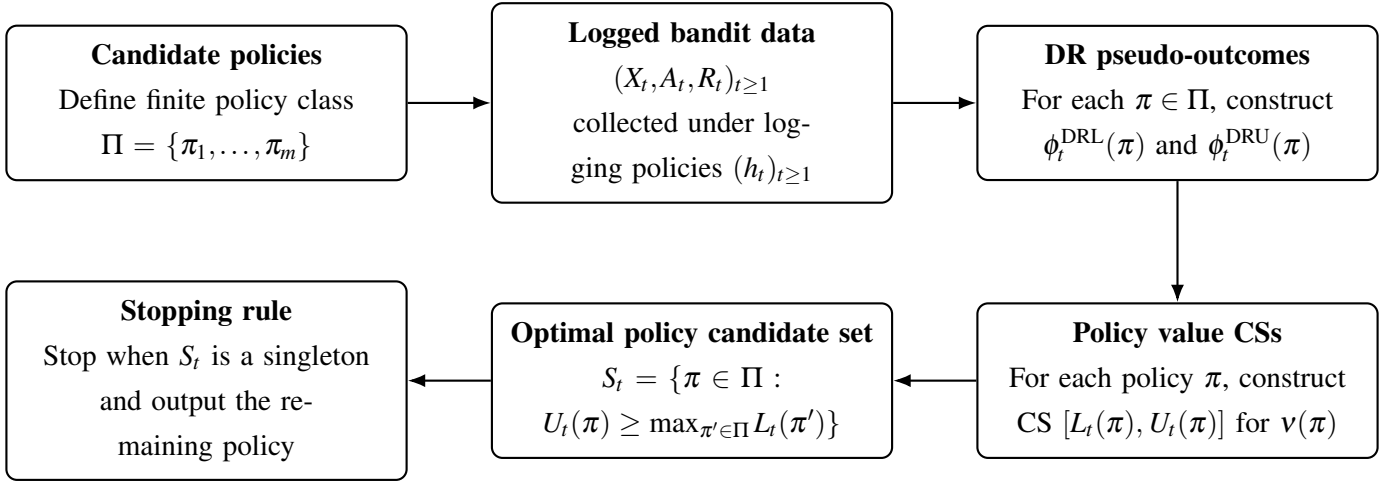

This formulation of an optimal policy candidate set extends naturally to related identification tasks. For example, by modifying the elimination criterion, the analyst can target Top-\(M\) policy identification or \(\epsilon\)-optimal policy sets, analogous to Top-\(M\) and \(\epsilon\)-optimal objectives in the best-arm identification literature \citep{bubeck_multiple_2012, locatelli_optimal_2016, mason_finding_2020, jun_anytime_2016}. We do not develop these extensions further here, but highlight that the framework naturally accommodates them.

\subsection{Stopping time and sample complexity}

We now turn to the case in which the optimal policy is unique, so that \(\Pi^\star=\pi^\star\). Theorem~\ref{thm:optimal-policy-set} guarantees that \(S_t\) contains \(\pi^\star\) uniformly over time with probability at least \(1-\alpha\), while progressively eliminating suboptimal policies as data accrue. Intuitively, once the underlying policy value CSs have shrunk enough to separate \(\pi^\star\) from every competing policy, the set \(S_t\) will uniquely identify \(\pi^\star\). We next formalize the stopping time at which this occurs and derive a corresponding sample-complexity guarantee.
To derive this guarantee, we begin from the policy value CSs introduced in Proposition~\ref{prop:lil_cs}. Building on the time-uniform concentration arguments of \citet{howard_time-uniform_2021}, \citet{waudby-smith_anytime-valid_2024} show that these CSs shrink at the following rate:
\begin{equation}
U_t\!\left(\pi;\frac{\alpha}{m}\right)-L_t\!\left(\pi;\frac{\alpha}{m}\right)
=
O\!\left(
\frac{\sqrt{\overline V_t(\pi)\bigl(\log\log \overline V_t(\pi)+\log m\bigr)}}{t}
\right)
\label{eq:lil-width-rate-general}
\end{equation}
uniformly across \(\Pi\).

To obtain a tractable sample-complexity guarantee for the stopping time, we impose the following regularity condition on the variance process.

\begin{assumption}[Linear variance growth]
\label{ass:linear-variance-growth}
For each policy \(\pi\in\Pi\), the stabilized variance process satisfies
\(
\overline V_t(\pi)=O(t)
\)
almost surely.
\end{assumption}

\noindent Assumption~\ref{ass:linear-variance-growth} rules out settings in which importance-weighted pseudo-outcomes become increasingly unstable over time, causing the variance process to grow faster than linearly. A sufficient condition is strict overlap: if there exists \(\eta>0\) such that \(h_t(a\mid x)\ge \eta\) for all \(t,x,a\), then the importance weights are uniformly bounded and the variance process grows at most linearly. Under Assumption~\ref{ass:linear-variance-growth}, \eqref{eq:lil-width-rate-general} simplifies to
\begin{equation}
U_t\!\left(\pi;\frac{\alpha}{m}\right)-L_t\!\left(\pi;\frac{\alpha}{m}\right)
=
O\!\left(
\sqrt{\frac{\log\log t+\log m}{t}}
\right).
\label{eq:lil-width-rate-simplified}
\end{equation}

For each suboptimal policy \(\pi\neq\pi^\star\), define its suboptimality gap
\(
\Delta_\pi := \nu(\pi^\star)-\nu(\pi)
\),
and let
\(
\Delta_{\min} := \min_{\pi\neq\pi^\star}\Delta_\pi
\) denote the minimum suboptimality gap over \(\Pi\).
The following theorem defines the stopping time at which \(S_t\) uniquely identifies \(\pi^\star\) and gives its sample-complexity rate.

\begin{theorem}[Stopping-time and sample complexity]
\label{thm:stopping-time}
Suppose that \(\Pi\) is a fixed finite policy class with \(m=|\Pi|\), and that \(\pi^\star\) is a unique optimal policy. Under Assumption~\ref{ass:linear-variance-growth}, fix a probability threshold \(\alpha \in (0, 1)\). Define the stopping time
\[
\tau
:=
\inf\left\{
t\ge1:
\exists \pi' \in \Pi
\text{ such that }
L_t\!\left(\pi';\frac{\alpha}{m}\right)
>
\max_{\pi \neq \pi'}
U_t\!\left(\pi;\frac{\alpha}{m}\right)
\right\}.
\]
Then, with probability at least \(1-\alpha\), \(\tau<\infty\) and
\(
S_\tau = \{\pi^\star\}.
\)
Moreover,
\[
\tau
=
O\!\left(
\frac{\log m + \log\log(1/\Delta_{\min})}{\Delta_{\min}^2}
\right)
\qquad
\text{as } \Delta_{\min} \downarrow 0.
\]
\end{theorem}

\noindent Intuitively, uniquely identifying \(\pi^\star\) becomes more difficult as the size of the policy class \(\Pi\) grows, but the dominant source of difficulty is how close the best policy is to its nearest competitor.

We provide a full proof in Appendix~\ref{app:stopping-time-proof}, but the argument is straightforward at a high level. On the simultaneous coverage event from \eqref{eq:simultaneous-cs-coverage}, each policy's CS contains its true value at all times. Under the width rate in \eqref{eq:lil-width-rate-simplified}, the confidence bands shrink uniformly at rate
\[
O\!\left(\sqrt{\frac{\log\log t+\log m}{t}}\right).
\]
Once the relevant CS widths become smaller than \(\Delta_{\min}\), the lower confidence bound of the optimal policy must exceed the upper confidence bound of every suboptimal policy. At that time, every suboptimal policy is eliminated and the set \(S_t\) collapses to \(\{\pi^\star\}\). Solving the inequality
\[
\sqrt{\frac{\log\log t+\log m}{t}} \lesssim \Delta_{\min}
\]
for \(t\) yields the stated rate
\[
\tau
=
O\!\left(
\frac{\log m + \log\log(1/\Delta_{\min})}{\Delta_{\min}^2}
\right).
\]

Together, Theorems~\ref{thm:optimal-policy-set} and~\ref{thm:stopping-time} provide a complementary pair of guarantees. Theorem~\ref{thm:optimal-policy-set} ensures anytime-valid coverage of the true optimal policy set, while Theorem~\ref{thm:stopping-time} gives a corresponding sample-complexity rate for unique identification of the optimal policy.

\section{Empirical results}

We now present three empirical demonstrations of our framework. First, we illustrate how anytime-valid optimal policy identification supports sequential monitoring of both the optimal policy candidate set and the underlying policy values. Second, we show that, relative to a corresponding fixed-\(N\) design, the anytime-valid procedure can deliver substantial sample-efficiency gains through early stopping. Third, we apply the method to the real-world experiment of \citet{offer-westort_battling_2024}, showing that it corroborates the study's conclusions while providing a sequential view of how those conclusions emerge as data accrue.

\subsection{Illustrative example}
\label{sec:illustrative_example}

To illustrate how the anytime-valid procedure operates in practice, we begin with a simple synthetic contextual bandit setting. Contexts are generated \(X_t \sim \mathrm{Uniform}(0,1)^3\). Actions \(A_t \in \{0, 1\}\) are assigned by a logistic logging policy with treatment propensity
\[
h_t(A_t=1\mid X_t)=\mathrm{clip}\!\left(\sigma(w^\top X_t),\,0.10,\,0.90\right),
\]
where \(\sigma\) is the logistic sigmoid function, \(w=[0.346,\,0.822,\,0.331]\) is a fixed coefficient vector, and propensities are clipped to remain between \(0.10\) and \(0.90\).
Rewards are Beta-distributed with conditional mean \(\mu(X_t,A_t)\):
\[
R_t \mid X_t, A_t
\sim
\mathrm{Beta}\!\left(\mu(X_t,A_t),\,1-\mu(X_t,A_t)\right),
\]
where
\[
\mu(X_t,A_t)
=
\beta_{A_t} + 0.1(X_{t1}+X_{t2}+X_{t3}),
\]
with \(\beta_0 = 0.25\) and \(\beta_1 = 0.55\).

We consider a class \(\Pi\) of six candidate policies, consisting of the logging policy and five additional target policies. The optimal target policy \(\pi^\star\) always assigns \(A_t=1\), and the four suboptimal policies are chosen so that their suboptimality gaps
\(
\Delta_\pi
\)
are \(0.05\), \(0.06\), \(0.07\), and \(0.08\), respectively.
We run the experiment for a horizon of \(T=5{,}000\) observations, and instantiate the CSs of Proposition~\ref{prop:lil_cs} with truncation level \(k=0\), \(\widehat\xi_0 = \tfrac{1}{2(1+k)}\), and an ordinary least squares reward predictor \(\widehat r_t(x,a)\).

Figure~\ref{fig:illustrative_example} demonstrates how our framework facilitates simultaneous inference on both the optimal policy \(\pi^\star\) and the underlying policy values \(\nu(\pi)\). The top panel displays the progressive evolution of \(S_t\) as data accumulate: clearly suboptimal policies are eliminated early, while the optimal policy \(\pi^\star\) remains in \(S_t\) throughout and is ultimately uniquely identified. The bottom panel shows the corresponding CSs for the policy values, which converge toward the true values as sample size grows.

\begin{figure}[!htbp]
    \centering
    \includegraphics[width=0.8\textwidth]{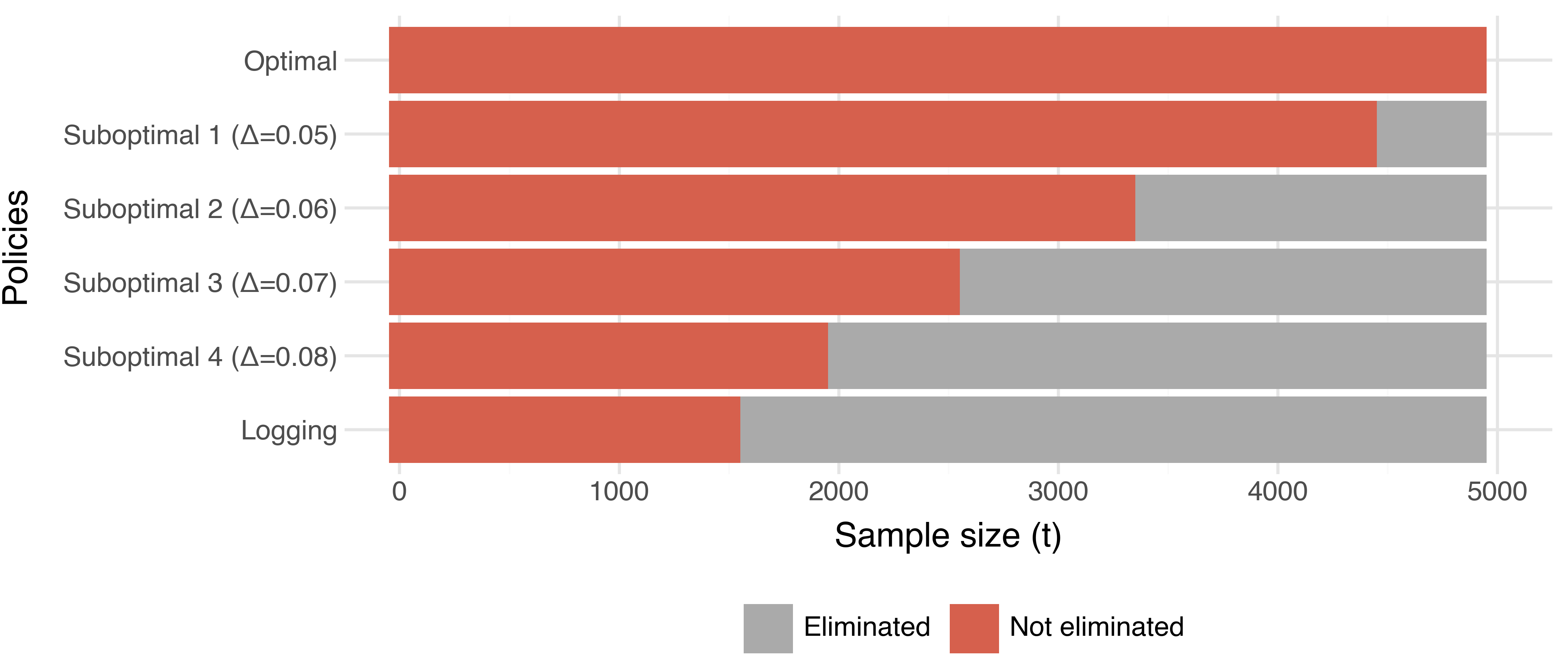}

    \vspace{0.8em}

    \includegraphics[width=0.7\textwidth]{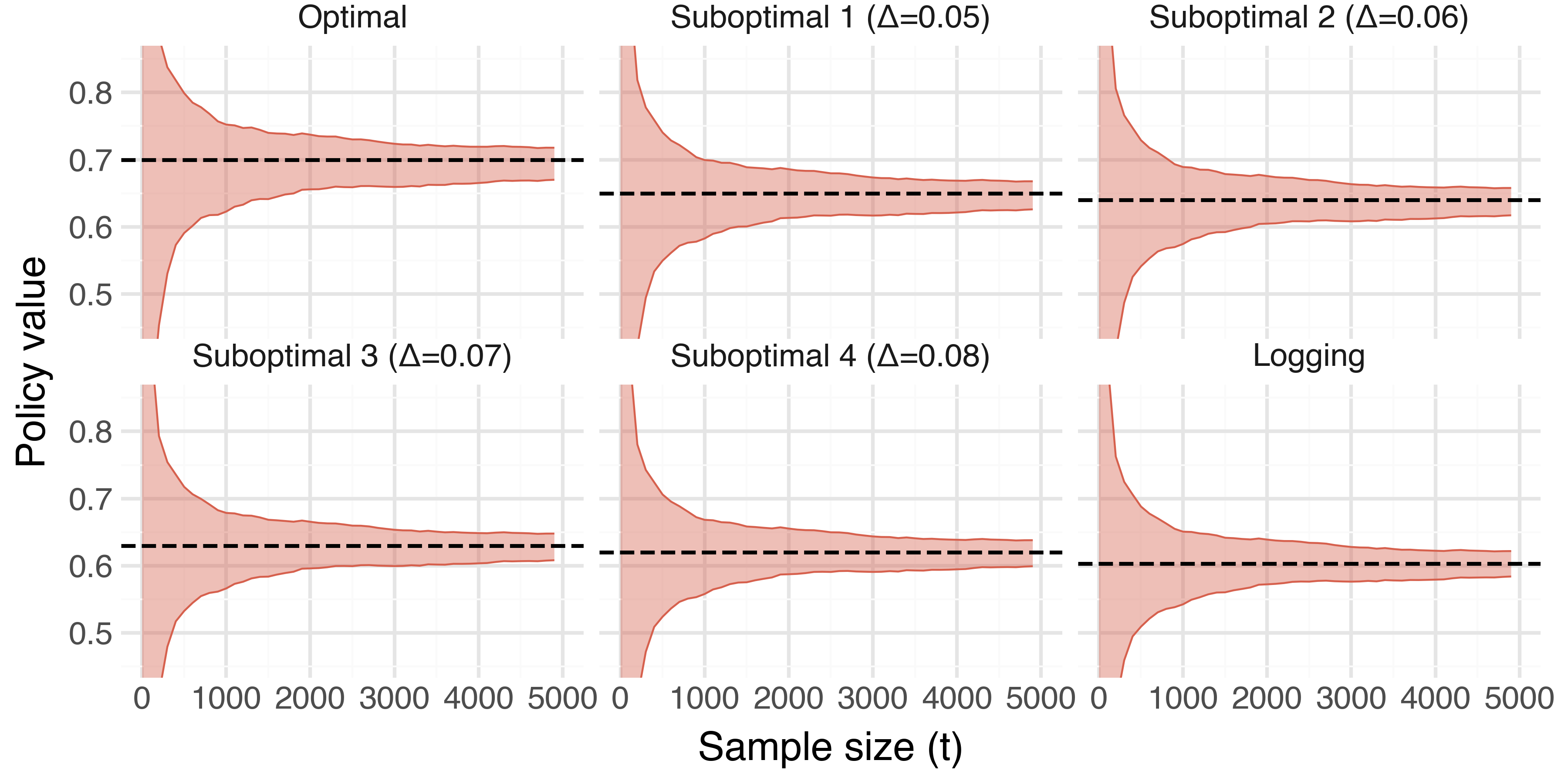}

    \caption{\textbf{Anytime-valid optimal policy identification.} \textit{Top panel}: progressive evolution of the optimal policy candidate set \(S_t\) over \(T=5{,}000\) observations. Clearly suboptimal policies are eliminated early, while the optimal policy \(\pi^\star\) remains in \(S_t\) throughout and is ultimately uniquely identified. \textit{Bottom panel}: corresponding CSs for the policy values \(\nu(\pi)\). Dashed horizontal lines mark the true policy values.}
    \label{fig:illustrative_example}
\end{figure}

\subsection{Anytime-valid sample-efficiency gains}
\label{sec:sample_efficiency}

Choosing an appropriate sample size is a fundamental challenge in experimental design. If a study collects too little data, the analyst may be unable to identify the optimal policy with adequate confidence. If it collects substantially more data than necessary, the analyst incurs unnecessary costs in time, effort, and resources. This issue is especially relevant in our setting, where data are generated by an external logging policy that may not be tuned to the analyst's specific inferential objective. We show that anytime-valid stopping can deliver substantial sample-efficiency gains when the logging policy collects more data than the optimal fixed sample size needed to identify \(\pi^\star\).

Using the same contextual data-generating process and hyperparameter choices as in Section~\ref{sec:illustrative_example}, with policy-class size \(m=10\), we vary the minimum suboptimality gap \(\Delta_{\min}\) to study settings with varying levels of identification difficulty. Specifically, we consider target gaps of \(\Delta_{\min} \in \{0.02, 0.05, 0.10, 0.15, 0.20\}\). For each target gap, we use the fixed-sample off-policy methods of \citet{waudby-smith_anytime-valid_2024} to compute the oracle fixed sample size \(N_{90}\) required to identify \(\pi^\star\) with \(90\%\) probability. If the true gap exceeds the target gap, then the fixed-sample design collects more data than necessary, whereas the anytime-valid procedure can stop as soon as \(\pi^\star\) is identified. We therefore simulate settings in which the true gap exceeds the target gap, over the range
\[
c := \frac{\text{true }\Delta_{\min}}{\text{target }\Delta_{\min}}
\in \{1, 1.25, 1.5, 1.75, 2, 2.5, 3\}.
\]
For each value of \(c\), we run \(1{,}000\) simulations and record the stopping time \(\tau\). We measure sample savings by \(1-\tau/N_{90}\), and report the mean sample savings in Figure~\ref{fig:sample_savings}.

\begin{figure}[!htbp]
    \centering
    \includegraphics[width=0.7\textwidth]{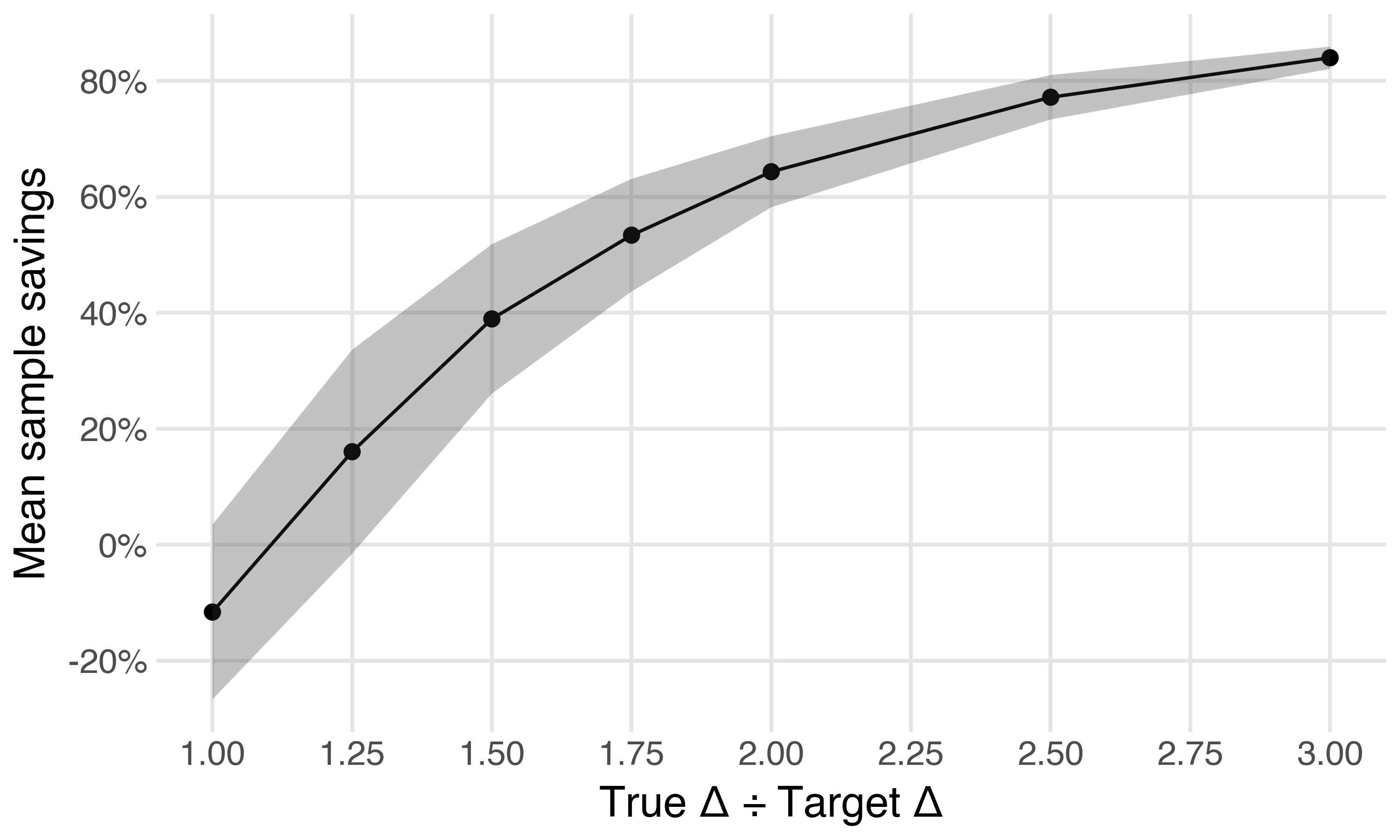}
    \caption{\textbf{Sample savings under target gap misspecification.} The x-axis shows the
    ratio of the true minimum suboptimality gap to the target gap used to determine the oracle fixed sample size; the y-axis
    shows mean sample savings \(1 - \mathbb{E}[\tau]/N_{90}\), averaged across target
    gap values, with \(\pm 1\) standard error ribbon.}
    \label{fig:sample_savings}
\end{figure}

Even when the true gap exceeds the target gap by only 25\%, the anytime-valid procedure can detect \(\pi^\star\) with approximately \(20\%\) less than the planned sample on average.
When the true gap is two to three times larger than the target gap, average savings approach \(60\%\) to \(80\%\).
Concretely, suppose that a study is planned to identify the optimal policy under a target gap of \(\Delta_{\min} = 0.02\). If the true gap is \(\Delta_{\min}=0.04\), the analyst can expect to recover, on average, more than 60\% of the planned sample size.

\subsection{Application: reducing the spread of misinformation online}
\label{sec:infodemic}

Finally, we apply the framework to the adaptive experiment of \citet{offer-westort_battling_2024}, who evaluated eight interventions for reducing the spread of misinformation among social media users in Kenya and Nigeria (\(N = 15{,}292\)). We consider a policy class of eight policies, each of which assigns all users to a single intervention: control, accuracy, deliberation, emotion suppression, pledge, AfricaCheck tips, Facebook tips, or video training. Data were collected using a contextual bandit design. The original study found that the accuracy nudge and Facebook tips policies were the most effective at reducing sharing of misinformation.

\begin{figure}[!htbp]
    \centering
    \includegraphics[width=0.7\textwidth]{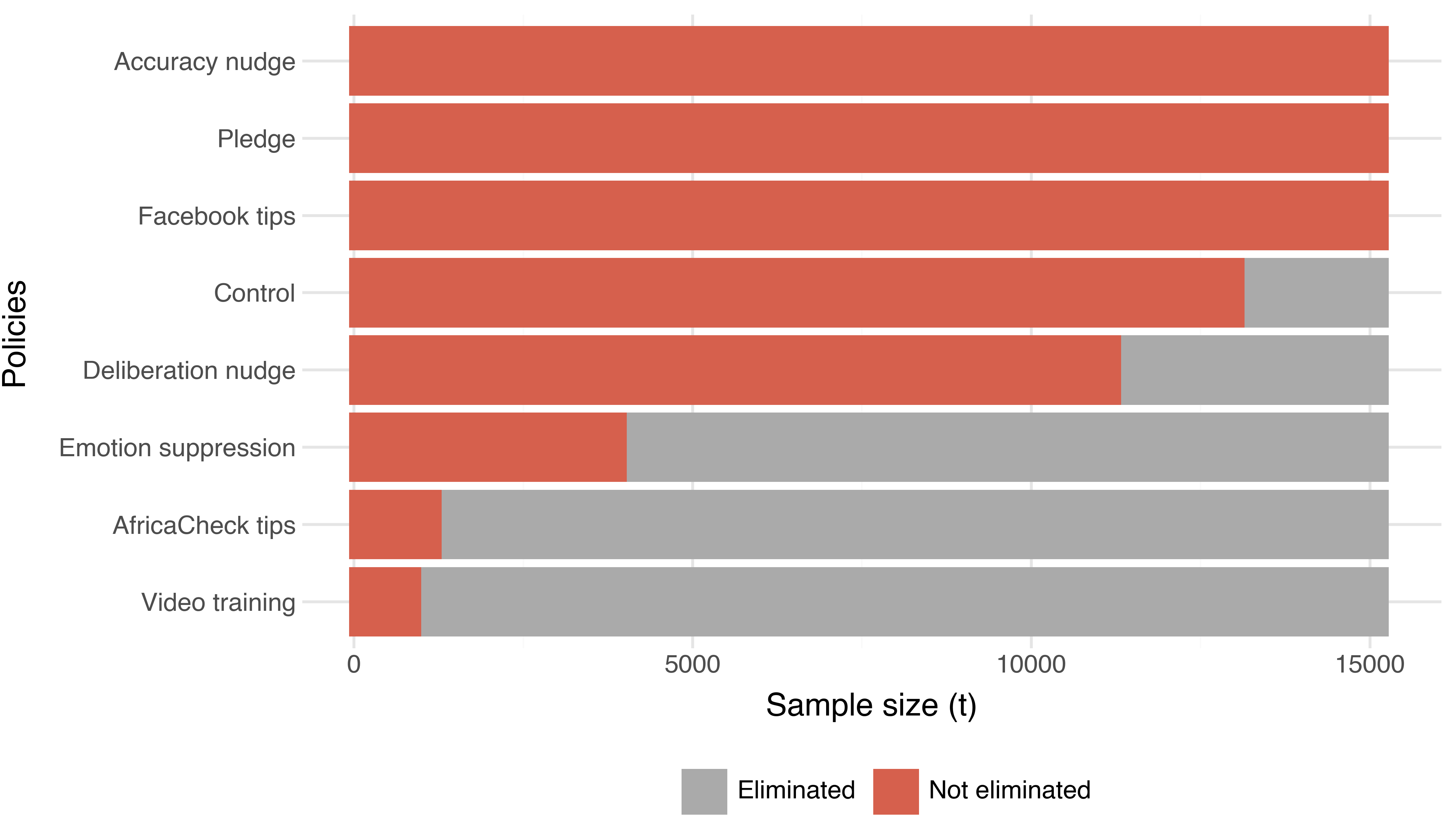}
    \caption{\textbf{Anytime-valid optimal policy identification in the infodemic study.}
    The heatmap displays \(S_t\), the candidate set of optimal policies, over
    \(N = 15{,}292\) observations. Policies surviving to the end of the study (never
    eliminated) are shown at the top; remaining policies are ordered by their elimination
    time, latest to earliest.}
    \label{fig:infodemic}
\end{figure}

Figure~\ref{fig:infodemic} shows how evidence on the optimal policy set \(S_t\) accumulates sequentially. The video training and AfricaCheck tips policies are eliminated very early, at approximately \(t = 1{,}000\), well before a tenth of the data have been collected. The emotion suppression and deliberation nudge policies follow at \(t = 4{,}000\) and \(t = 11{,}000\), and the control policy is eliminated near the end of the study at \(t = 13{,}000\). Three policies are never eliminated: the accuracy nudge policy, the pledge policy, and the Facebook tips policy. This mirrors the original study's conclusion that the accuracy nudge and Facebook tips policies are the top performers in this candidate class.

The pledge policy is more ambiguous: it performs well enough that it cannot be eliminated from \(S_t\) on the basis of the available data, while the five remaining policies are decisively eliminated. Critically, much of this information would have been available long before the study concluded. The worst-performing policies are identified as suboptimal with only a fraction of the final sample, illustrating how anytime-valid monitoring can provide ongoing, actionable guidance throughout a study rather than only at its conclusion.

\section{Broader impacts and limitations}

This work provides practitioners with tools for sequential policy evaluation in settings where data are generated by an externally determined logging policy. This is especially relevant in real-world policy evaluation problems, where analysts often face limited time, resources, and operational flexibility. By enabling valid monitoring as evidence accrues, the framework can help analysts identify clearly suboptimal policies earlier and make more transparent decisions about whether to continue data collection.

Two limitations are worth noting. First, the framework assumes time-invariant policy values. This assumption may fail in settings where the respondent population changes over time or where treatment effects evolve with repeated exposure, for example due to novelty or fatigue effects. Extending the framework to such non-stationary settings is a natural direction for future work.

Second, the inferential guarantees require overlap: each target policy must be absolutely continuous with respect to the logging policy. This allows deterministic target policies, but rules out settings where the logging policy assigns zero probability to actions that a target policy would take. This condition can be violated in practice under highly concentrated or deterministic bandit algorithms, such as UCB-style algorithms. Although this limits compatibility with some adaptive designs, the framework still applies to a broad class of stochastic logging policies commonly used in practice.

\section{Summary}

We develop an anytime-valid framework for optimal policy identification from externally logged contextual bandit data. We construct a time-indexed candidate set \(S_t\) that retains the optimal policy with probability at least \(1-\alpha\), uniformly over all sample sizes. This guarantee permits sequential elimination of clearly suboptimal policies and stopping at arbitrary data-dependent times without invalidating inference. When the optimal policy is unique, we characterize the stopping time at which it is uniquely identified and derive a corresponding sample-complexity bound scaling as \(O\!\bigl((\log m + \log\log(1/\Delta_{\min}))/\Delta_{\min}^2\bigr)\). Empirically, this framework can deliver meaningful sample-efficiency gains over fixed-\(N\) designs and, in a large real-world adaptive experiment, corroborates the original study's conclusions while showing how evidence on the optimal policy accumulates over time.

\section*{Data and replication materials}

A full replication package, including code and data to reproduce all simulations and empirical analyses, is available at \url{https://github.com/dmolitor/av-policy-selection}.

\newpage

\bibliography{references}

\newpage

\appendix

\section{Proof of Theorem~\ref{thm:optimal-policy-set}}
\label{app:proof-optimal-policy-set}

\begin{proof}
For each \(t \ge 1\), define
\[
E(t) := \left\{\nu(\pi) \in \left[L_t\left(\pi; \frac{\alpha}{m}\right), U_t\left(\pi; \frac{\alpha}{m}\right)\right] : \forall \pi \in \Pi \right\},
\]
and let
\[
E := \bigcap_{t\ge1} E(t).
\]
As established in the main text \eqref{eq:simultaneous-cs-coverage}, \(\PP(E) \ge 1-\alpha\). Now fix any \(\pi^\dagger \in \Pi^\star\) and any \(t \ge 1\). If \(E(t)\) holds true, we have that
\[
L_t\left(\pi'; \frac{\alpha}{m}\right) \le \nu(\pi') \le \nu(\pi^\dagger)
\qquad \text{for all } \pi' \in \Pi,
\]
since \(\pi^\dagger\) is optimal. Therefore,
\[
\max_{\pi'\in\Pi} L_t\left(\pi'; \frac{\alpha}{m}\right) \le \nu(\pi^\dagger) \le U_t\left(\pi^\dagger; \frac{\alpha}{m}\right).
\]
Combining these gives
\[
\max_{\pi'\in\Pi} L_t\left(\pi'; \frac{\alpha}{m}\right)
\le
U_t\left(\pi^\dagger; \frac{\alpha}{m}\right),
\]
so \(\pi^\dagger \in S_t\). Since \(\pi^\dagger \in \Pi^\star\) was arbitrary, it follows that when \(E(t)\) holds, \(\Pi^\star \subseteq S_t\). Thus,
\[
E \implies \Pi^\star \subseteq S_t \ \ \forall t\ge1.
\]
Finally,
\[
\PP\left(\Pi^\star \subseteq S_t : \forall t \ge 1\right)
\ge
\PP(E)
\ge
1-\alpha.
\]
\end{proof}

\section{Proof of Theorem~\ref{thm:stopping-time}}
\label{app:stopping-time-proof}

Before diving into the proof of Theorem~\ref{thm:stopping-time} we will state and prove a small but crucial lemma.

\subsection{Intermediate lemma}

\begin{lemma}[]
Suppose \(t \in \mathbb{R}: t \ge C\) for some arbitrary \(C > 1\). Additionally, suppose we have constants \(A, K \ge e\). Define \(L(x):=\log\!\bigl(\log(x)\bigr)\), \(X := A + L(K) + L(A) + C\), and \(T := 3KXC\). 
It follows that
\[
T > K\bigl(A+L(T)\bigr).
\]
Consequently,
\[
t^\star := \inf\left\{ t \in [C, \infty): t > K\bigl(A+L(t)\bigr) \right\} \le T.
\]
\label{lemma:critical_lemma}
\end{lemma}

\begin{proof}
We begin by establishing that, since \(L(K), \ L(A)\ge \log\log(e)=0\), we have \(X \ge e + 1\). Similarly, since \(X \ge e + 1\) and \(K \ge e\), it follows that \(T \ge 3(e^2 + e)\). Next, we will show that \(L(T)\le X+2\). First, note that by construction \(\log \log K := L(K) \le X\) and \(C \le X\) which implies that \(\log K, \ \log C \le e^X\). Next, since \(X \ge e+1\), it follows that \(\log X \le e^X\), and \(\log 3 \le e^{-(e + 1)}e^X\log 3\), so
\begin{align*}
    \log(T) &=\log3 + \log K + \log X + \log C \\
    &\le e^{-(e + 1)}e^X\log 3 +3e^X \\
    &= (e^{-(e + 1)}\log 3 + 3)e^X.
\end{align*}
It follows that
\begin{align*}
    L(T) &:= \log\log T \\
    &= \log(\log 3 + \log K + \log X + \log C) \\
    &\le \log((e^{-(e + 1)}\log 3 + 3)e^X) \\
    &= \log(e^{-(e + 1)}\log 3 + 3) + X \\
    &< X + 2.
\end{align*}
So, \(L(T) \le X + 2\) and as a result \(K(A + L(T)) \le K(A + X + 2)\). Finally, since by construction \(A \le X\) and \(X \ge 3\), we have \(2 \le \frac{2X}{3}\) and thus 
\begin{align*}
    A + X + 2 &\le X + X + \frac{2X}{3} \\
    &<3X.
\end{align*}
Combining these pieces, we see that
\begin{align*}
    K(A + L(T)) &\le K(A + X + 2) \\
    &<3KX \\
    &< 3KXC \\
    &= T.
\end{align*}
Since by construction \(T > C\), then \(t^\star \le T\).
\end{proof}

\subsection{Proof of Theorem~\ref{thm:stopping-time}}

We will now state the proof of Theorem~\ref{thm:stopping-time}.

\begin{proof}
Suppose that the simultaneous coverage event holds:
\[
E := \left\{\nu(\pi) \in \left[L_t\left(\pi; \frac{\alpha}{m}\right), U_t\left(\pi; \frac{\alpha}{m}\right)\right]: \forall t \ge 1, \forall \pi \in \Pi \right\}.
\]
By \eqref{eq:lil-width-rate-simplified} and the finiteness of \(\Pi\), there exist constants \(D>0\) and \(t_0 > 1\) such that, for every policy \(\pi \in \Pi\) and all \(t \ge t_0\),
\[
U_t\left(\pi; \frac{\alpha}{m}\right) - \nu(\pi), \ \nu(\pi) - L_t\left(\pi; \frac{\alpha}{m}\right)
\le
D\sqrt{\frac{\log\log t+\log m}{t}}.
\]
It immediately follows that, for every suboptimal policy \(\pi \neq \pi^\star\) and all \(t \ge t_0\),
\begin{align*}
L_t(\pi^\star)
-
U_t(\pi)
&\ge
\nu(\pi^\star)
-
D\sqrt{\frac{\log\log t+\log m}{t}}
-
\left(
\nu(\pi)+
D\sqrt{\frac{\log\log t+\log m}{t}}
\right) \\
&=
\Delta_\pi - 2D\sqrt{\frac{\log\log t+\log m}{t}}.
\end{align*}
Therefore, for any \(t \ge t_0\),
\[
2D\sqrt{\frac{\log\log t+\log m}{t}} < \Delta_{\min}
\implies
L_t(\pi^\star)\ >\ \max_{\pi \neq \pi^\star} U_t(\pi),
\]
and thus \(\tau \le t\). Squaring and rearranging \(2D\sqrt{\frac{\log\log t+\log m}{t}} < \Delta_{\min}\) gives us
\[
t>\frac{4D^2}{\Delta_{\min}^2}\bigl(\log m+\log\log t\bigr).
\]
Thus, it suffices to find
\[
t' = \inf\left\{ t \in [t_0, \infty): t > \frac{4D^2}{\Delta_{\min}^2}\bigl(\log m + \log\log t\bigr) \right\}.
\]
It immediately follows that it suffices to find
\[
t^\star = \inf\left\{ t \in [t_0, \infty): t > \max\left\{e, \frac{4D^2}{\Delta_{\min}^2}\right\}\bigl(\max\{e, \log m\} + \log\log t\bigr) \right\},
\]
since
\[
\max\left\{e, \frac{4D^2}{\Delta_{\min}^2}\right\}\bigl(\max\{e, \log m\} + \log\log t\bigr)
\ge
\frac{4D^2}{\Delta_{\min}^2}\bigl(\log m + \log\log t\bigr).
\]
Hence \(\tau \le t^\star\). By Lemma~\ref{lemma:critical_lemma}, it also follows that
\begin{align}
\tau
&\le t^\star < T \notag\\
&:= 3t_0\max\left\{e, \frac{4D^2}{\Delta_{\min}^2}\right\}
\left(
\max\{e, \log m\}
+ \log \log \max\left\{e, \frac{4D^2}{\Delta_{\min}^2}\right\}
+ \log \log \max\{e, \log m\}
+ t_0
\right) \notag\\
&\le F\left(
3t_0\frac{4D^2}{\Delta_{\min}^2}
\left(
\log m
+ \log \log \frac{4D^2}{\Delta_{\min}^2}
+ \log \log \max\{e, \log m\}
+ t_0
\right)
\right) \notag\\
&< \infty.
\label{eq:tau_lt_infty}
\end{align}
for a sufficiently large constant \(F>0\) depending only on \(D\) and \(t_0\). Therefore, there exists a constant \(M>0\) such that
\[
\tau
\le
M
\frac{
\log m
+ \log \log \left(1/\Delta_{\min}^2\right)
+ \log \log \max\{e, \log m\}
+ 1
}{\Delta_{\min}^2}.
\]
In particular, since \(\Pi\) is fixed and finite, there exist constants
\(\Delta_0 \in (0,1)\) and \(N>0\) such that for all
\(0<\Delta_{\min}\le \Delta_0\),
\[
\tau \le N \frac{\log m + \log\log(1/\Delta_{\min})}{\Delta_{\min}^2}.
\]
That is,
\[
\tau = O\!\left(\frac{\log m + \log\log(1/\Delta_{\min})}{\Delta_{\min}^2}\right)
\qquad \text{as } \Delta_{\min}\downarrow 0.
\]

\noindent Finally, \(\tau < \infty\) by \eqref{eq:tau_lt_infty}, which implies by definition of \(S_t\) that \(S_\tau = \{\pi^\star\}\). By Equation~\ref{eq:simultaneous-cs-coverage}, \(\PP(E) \ge 1-\alpha\), and thus all demonstrated results also hold with probability at least \(1-\alpha\).
\end{proof}

\end{document}